%% file: ga.tex
\documentclass{llncs}
\usepackage[pdftex]{graphicx}
\usepackage{graphics,color}
\usepackage{amsmath}
\usepackage{amssymb}
\usepackage{paralist}
\usepackage{verbatim}
\usepackage{array}
\usepackage{wrapfig}
\usepackage{subcaption}
\usepackage{todonotes}
\definecolor{gray}{rgb}{.7, .7, .7}

\newcommand{\set}[1]{\{{#1}\}}

\newcommand{\gencauchy}[1]{\mathit{GenCauchy}({#1})}

\newcommand{\norm}[1]{\left|\!\left|{#1}\right|\!\right|}
\newcommand{\abs}[1]{\left|{#1}\right|}
\newcommand{\pr}[1]{\mathbf{Pr}[{{#1}}]}

\newcommand{\prpre}[1]{\mathbf{Pr}_{pre}[{{#1}}]}
\newcommand{\prpost}[1]{\mathbf{Pr}_{post}[{{#1}}]}
\newcommand{\noised}[1]{\mathcal{M}_{{#1}}}

\newcommand{\rndvar}[1]{\mathbf{{#1}}}

\newcommand{\fxpre}{f_X}

\newcommand{\fxprek}[1]{f_{{#1}}}

\newcommand{\fy}{f_Y}

\newcommand{\RR}{\mathbb{R}}
\newcommand{\ZZ}{\mathbb{Z}}
\newcommand{\NN}{\mathbb{N}}

\newcommand{\erf}[1]{\mathsf{erf}\left({#1}\right)}

\newcommand{\sensdata}{set of sensitive attributes}
\newcommand{\ub}[2]{UB^{{#1}}_{{#2}}}

\title{Interpreting Epsilon of Differential Privacy in Terms of Advantage in Guessing or Approximating Sensitive Attributes\thanks{This research was funded by the Air Force Research laboratory (AFRL) and Defense
Advanced Research Projects Agency (DARPA) under contract FA8750-16-C-0011. The
views expressed are those of the author(s) and do not reflect the official policy or position of
the Department of Defense or the U.S. Government.}}
\author{Peeter Laud and Alisa Pankova}
\institute{Cybernetica AS}

\begin{document}
\maketitle

\begin{abstract}
There are numerous methods of achieving $\epsilon$-differential privacy (DP). The question is what is the appropriate value of $\epsilon$, since there is no common agreement on a ``sufficiently small'' $\epsilon$, and its goodness depends on the query as well as the data. In this paper, we show how to compute $\epsilon$ that corresponds to $\delta$, defined as the adversary's advantage in probability of guessing some specific property of the output. The attacker's goal can be stated as Boolean expression over guessing particular attributes, possibly within some precision. The attributes combined in this way should be independent. We assume that both the input and the output distributions have corresponding probability density functions, or probability mass functions.
\end{abstract}

\section{Introduction}

Computing queries on statistical databases inevitably leaks something about the underlying data, even if only an aggregated output is observed. Differential privacy~\cite{DBLP:conf/icalp/Dwork06} (DP) quantifies privacy losses coming from such queries, and estimates accumulation of the losses if several queries are made to the same data. Roughly speaking, DP says that, if two databases are sufficiently similar, then the attacker should not be able (up to certain extent, parametrized by $\epsilon \geq 0$) to distinguish between them by observing the query output. If the outputs are numeric, DP is commonly achieved by adding calibrated noise to them. The smaller $\epsilon$ is, the larger is the noise magnitude, and the higher is privacy.

While cryptographic security parameters are in general well-understandable, it is not so easy with DP. In general, there is no upper bound on $\epsilon$, and we cannot tell in advance how small $\epsilon$ is ``good enough''. Suppose that we have come up with some fixed $\epsilon_0$ to conceal the change in a numeric attribute of an input table (let us denote it $x$) by $1$. The number $1$ is a very abstract value, and we need to know what it means in reality. For example, if $x$ is population size, then $1$ is a very small change that may be not worth of hiding, and if $x \in [0,1]$ is a probability, then $1$ would cover the entire space of inputs $x$. In the first case, we may want to enhance the privacy and require that e.g. change $\Delta = 1000$ should be untrackable. With respect to the new metric, we get exactly the same level of differential privacy for $\epsilon = \epsilon_0 \cdot 1000 = 0.001 \cdot \epsilon_0$. We see that we cannot define in advance ``sufficiently good'' $\epsilon$.

Intuitively, in the last example we have linked goodness of $\epsilon$ to the underlying metric, so we should consider $\epsilon$ together with the metric to define its goodness. We want to convert these two quantities to a single estimate that would be more standard and easier to interpret. In this paper, we show how to relate the metric and the $\epsilon$ to attacker's advantage in guessing a particular attribute.

\section{Related Work}

Differential privacy was first introduced by Dwork~\cite{DBLP:conf/icalp/Dwork06}. Numerous follow-up papers show how to achieve it for certain types of queries and database metric~\cite{DBLP:conf/pet/ChatzikokolakisABP13,DBLP:conf/birthday/ElSalamounyCP14,DBLP:conf/icdcit/ChatzikokolakisPS15,DBLP:journals/tods/KiferM14,DBLP:conf/sigmod/HeMD14,DBLP:conf/popl/EbadiSS15}. PINQ~\cite{DBLP:conf/sigmod/McSherry09} is probably the most famous example of a worked-out framework that provides privacy-preserving replies to database queries.

This work is mainly motivated by~\cite{DBLP:conf/isw/LeeC11}, which is one of the few papers mentioning the way of choosing appropriate $\epsilon$, relating it to the difference in attacker's prior (before observing the output) and posterior (after observing the output) knowledge about the input data. Since $\epsilon$ of DP may be difficult to interpret, there also exist alternative definitions of DP like Pufferfish privacy~\cite{pufferfish} or Differential identifiability~\cite{DBLP:conf/kdd/LeeC12}. In this paper, similarly to~\cite{DBLP:conf/isw/LeeC11}, we are staying within standard definition of DP and do not consider alternative definitions.

The results of~\cite{DBLP:conf/isw/LeeC11} have been briefly mentioned in~\cite{dersens}, being generalized from Laplace noise to arbitrary DP mechanisms, showing its potential usefulness to non-standard database metric. In this paper, we formalize these initial results, extend them, and provide technical details of applying it in practice. In general, we can use a system that supports non-standard metrics (like~\cite{dersens}) to estimate the amount of noise. For some particular cases like filtering by a discrete attribute, we get standard metric, so we can utilize obtained $\epsilon$ e.g. directly in PINQ system.

\section{Guessing Advantage and Differential Privacy}

\subsection{Definition of Guessing Advantage}\label{sec:ga}
We start from the attacker model of~\cite{DBLP:conf/isw/LeeC11}, where the attacker advantage is defined as the difference between its prior and posterior beliefs on the property that he is guessing. While~\cite{DBLP:conf/isw/LeeC11} is based on Laplace noise distribution, the analysis developed in~\cite{dersens} is based on generalized Cauchy distribution, for which it is more difficult to compute precise bounds. In this paper we work with more general bounds, which can be derived directly from the definition of differential privacy, without relying on the specific privacy mechanism.

First, let us give some definitions that we will use in this chapter.
\begin{definition}[Differential privacy]\label{def:dp}
Let $X$ be a metric space. Given $\epsilon \geq 0$, a mechanism $\noised{q}$ is $\epsilon$-differentially private if, for any $x,x' \in X$ such that $d(x,x') \leq d$,  and for any subset $Y \subseteq \noised{q}(X)$ of outputs, we have
\[\pr{\noised{q}(x) \in Y} \leq e^{d\epsilon}\cdot\pr{\noised{q}(x') \in Y}\enspace.\]
\end{definition}

\begin{definition}[Differential privacy w.r.t. component]\label{def:dpcomp}
Let $X = (X_1,\ldots,X_n)$ be a metric space. Given $\epsilon \geq 0$, a mechanism $\noised{q}$ is $\epsilon$-differentially private w.r.t. the component $X_j$ if, for any $x,x' \in X$ such that $d(x_j,x'_j) \leq d$, $x_i = x'_i$ for all $i \neq j$, and for any subset $Y \subseteq \noised{q}(X)$ of outputs, we have
\[\pr{\noised{q}(x) \in Y} \leq e^{d\epsilon}\cdot\pr{\noised{q}(x') \in Y}\enspace.\]
\end{definition}

While Def.~\ref{def:dp} is a general definition of differential privacy, Def.~\ref{def:dpcomp} is its particular instantiation on multidimensional data.

We are now ready to define the guessing advantage. Let $x \in X$ be the input. The attacker has a goal $g:\mathcal{S} \to \set{0,1}^{*}$, which defines the information that he \emph{wants to learn} about $x$. There is a function $e:\mathcal{S} \to \set{0,1}^{*}$ defining the information that the attacker \emph{already knows} about $x$. We want to know how the distribution of $g(X)$ changes after the attacker in addition gets the output $q(x)$ of a (differentially private) query $q:\mathcal{S} \to Y$.

The security definition is related to the difference between posterior (with observation $q(x)$) and prior (without observation $q(x)$) probabilities of guessing $g(X) = g(x)$, assuming that the extra information $e(x)$ is included as a condition of both prior and posterior probabilities.

\begin{definition}\label{def:gadisc}[Guessing advantage (discrete)]
Let $x \in X$ be the data instance, $\noised{q}$ the observation, $g$ the attacker goal, and $e$ the extra information. We say that advantage of guessing $g(x)$ is at most $\delta$ if, for any algorithm $A$,
\[\abs{\pr{A (\noised{q}(x),e(x)) = g(x)\ |\ Y} - \pr{A (e(x)) = g(x)\ |\ Y}} \leq \delta\]
where $Y = (\noised{q}(X) = \noised{q}(x)) \wedge (e(X) = e(x))$.
\end{definition}

For continuous data, the probability mass of each single point may be $0$. It is possible that $\pr{g(X) = g(x)} = 0$ for all $x$, so Def.~\ref{def:gadisc} does not make any sense. Indeed, in practice it can be equally bad if the attacker learns some geographical location "precisely" or "close enough", so we need to introduce the notion of precision. The attacker wants to come up with a point $x'$ such that $d(x,x') \leq r$ for a sufficiently small $r$.

\begin{definition}\label{def:gacont}[Guessing advantage (continuous)]
Let $x \in X$ be the data instance, $\noised{q}$ the observation, $g$ the attacker goal, and $e$ the extra information. Let $\mathcal{B}(x,r) = \set{x' | d(x,x') \leq r}$.  We say that advantage of guessing $g(x)$ with precision $r$ is at most $\delta$ if, for any algorithm $A$,
\[\abs{\pr{A (\noised{q}(x),e(x)) \in \mathcal{B}(g(x),r)\ |\ Y} - \pr{A (e(x)) \in \mathcal{B}(g(x),r)\ |\ Y}} \leq \delta\]
where $Y = (\noised{q}(X) = \noised{q}(x)) \wedge (e(X) = e(x))$.
\end{definition}

While the algorithm $A$ can in principle be arbitrary, we consider an attacker who knows the distribution of inputs, i.e. attacker's prior belief is equivalent to the actual distribution of inputs. If he only gets the extra knowledge $e(x)$, he estimates the probability weight of "sufficiently correct answers" $\mathcal{B}(g(x),r)$ as $\prpre{\mathcal{B}(g(x),r)} := \pr{X \in \mathcal{B}(g(x),r)\ |\ e(X) = e(x)}$, i.e. the \emph{prior} probability. If he gets in addition the output $\noised{q}(x)$, the confidence is $\prpost{\mathcal{B}(g(x),r)} := \prpre{\mathcal{B}(g(x),r)\ |\ \noised{q}(X) = \noised{q}(x), e(X) = e(x)}$, i.e. the \emph{posterior} probability.

\subsection{From Guessing Advantage to Epsilon}\label{sec:gadp}

Let the posterior belief of the adversary be expressed by the probability distribution $\prpost{\cdot}$. Let the initial distribution of $X$ be $\prpre{\cdot}$, and $\fxpre$ the corresponding probability density function (PDF), i.e. $\prpre{X'} = \int_{X'} \fxpre(x) dx$ for $X' \subseteq X$. Let $\fy$ be the PDF of the outputs of $\fy$. Note that existence of $\fxpre$ and $\fy$ is required for the further analysis. We want to compute an upper bound on $\prpost{X'}$ for some $X' \subseteq X$.

We have $\prpost{X'} := \pr{x \in X'\ |\ \noised{q}(x) = y}$, where $y$ is the output that the attacker observed. Let $\int_{A}$ denote a Lebesgue integral over a subset $A$ of some metric space. Similarly to~\cite{DBLP:conf/isw/LeeC11} and ~\cite{dersens}, we rewrite the posterior probability as
\begin{eqnarray*}
\prpost{X'} & = & \int_{X'} \fxpre(x'\ |\ \noised{q}(x') = y)\ dx'\\
& = & \int_{X'} \frac{\fy(y\ | x')\fxpre(x')}{\fy(y)}\ dx' = \frac{\int_{X'} \fy(y\ | x')\fxpre(x')\ dx'}{\int_{X} \fy(y\ | x)\ \fxpre(x)\ dx}\\
& = &  \frac{\int_{X'} \fy(y\ | x')\fxpre(x')\ dx'}{\int_{X'} \fy(y\ | x)\ \fxpre(x)\ dx + \int_{X\setminus X'} \fy(y\ | x)\ \fxpre(x)\ dx}\\
& = & \frac{1}{1 + \frac{\int_{X\setminus X'} \fy(y | x) \fxpre(x) dx}{\int_{X'} \fy(y | x') \fxpre(x') dx'}}\enspace.
\end{eqnarray*}

Let us assume that $R := \sup_{x \in X}{d(x,x')}$ exists. Differential privacy gives us $\frac{\pr{\noised{q}(x') \in Y}}{\pr{\noised{q}(x) \in Y}} \leq e^{\epsilon\cdot R}$ for all $Y$, and hence also $\frac{\fy(y | x')}{\fy(y | x)} \leq e^{\epsilon\cdot R}$. We get
\[\prpost{X'} \leq \frac{1}{1 + e^{-\epsilon R}\cdot\frac{\prpre{X\setminus X'}}{\prpre{X'}}}\enspace.\]
Hence, if we want $\prpost{X'} \leq \prpre{X'} + \delta$, we need to take
\[\epsilon \leq \epsilon_{lb} := \frac{-\ln\left(\frac{\prpre{X'}}{\prpre{X\setminus X'}} \cdot (\frac{1}{\delta + \prpre{X'}} - 1)\right)}{R}\enspace.\]

So far, we have shown that $\prpost{X'} \leq \prpre{X'} + \delta$. To satisfy Def.~\ref{def:ga}, we also need $\prpost{X'} \geq \prpre{X'} - \delta$, or $\prpost{X \setminus X'}) \leq \prpre{X \setminus X'} + \delta$, to show that the distribution has not changed much in overall compared to prior. The derivation of $\epsilon$ for the lower bound is analogous to the upper bound. Substituting $\prpre{X'}$ with $\prpre{X \setminus X'}$, we get a bound

\[\epsilon \leq \epsilon_{ub} := \frac{-\ln\left(\frac{\prpre{X\setminus X'}}{\prpre{X'}} \cdot (\frac{1}{\delta + \prpre{X\setminus X'}} - 1)\right)}{R}\enspace.\]

To satisfy $\abs{\prpost{X'} - \prpre{X'}} \leq \delta$, we eventually need to take

\[\epsilon = \min(\epsilon_{lb},\epsilon_{ub})\enspace.\]

Note that, if $p \leq 1 - p$, then $\epsilon_{lb} \leq \epsilon_{ub}$, and if $1 - p \leq p$, then $\epsilon_{ub} \leq \epsilon_{lb}$. Hence, we do not need to compute both bounds.

\paragraph{Constraining the search space}

The problem of this approach is that $R$ can be very large in practice, or even not exist. Note that we are essentially trying to prove that the elements of $X'$ are ``sufficiently indistinguishable'' from the elements of $X\setminus X'$. In practice, it may be sufficient to take just a subset of $\hat{X'} \subseteq X\setminus X'$ and show that it is difficult to distinguish $X'$ and $\hat{X'}$. We have
\begin{eqnarray*}
\prpost{X'} & = & \frac{1}{1 + \frac{\int_{X\setminus X'} \fy(y | x) \fxpre(x) dx}{\int_{X'} \fy(y | x') \fxpre(x') dx'}} \leq \frac{1}{1 + \frac{\int_{\hat{X'}} \fy(y | x) \fxpre(x) dx}{\int_{X'} \fy(y | x') \fxpre(x') dx'}}
\end{eqnarray*}
for any $\hat{X'} \subseteq X \setminus X'$. Let $\hat{X'_a} := \set{x\ |\ d(x,x') \leq a} \setminus X'$ for some $a \in \RR$. We get
\begin{equation}\label{eq:main}
\epsilon \leq \frac{-\ln\left(\frac{\prpre{X'}}{\prpre{\hat{X'_a}}} \cdot (\frac{1}{\delta + \prpre{X'}} - 1)\right)}{a}\enspace.
\end{equation}

If $R := \sup_{x \in X, x' \in X'}{d(x,x')}$ does not exist, then we may instead take $a$ such that $\pr{x\ |\ \forall x' \in X':\ d(x,x') \leq a} \approx 1$, which is useful e.g. when the input comes from normal distribution. It is possible that there are better candidates for $a$. A method to choose optimal $a$ for a uniformly distributed attribute is given in App.~\ref{app:choosinga}, but we do not have a generic approach for this. In practice, if there is $R$ such that $\pr{x\ |\ \forall x' \in X':\ d(x,x') \leq R} \approx 1$, since computing $\epsilon$ from $a$ is a cheap operation, we may sample several different values of $a$ from the interval $(0, R]$, and take one for which $\epsilon$ is the largest, i.e. noise is the smallest.

\section{Instantiation to databases}\label{sec:dbga}

As in~\cite{dersens}, we assume that a database is an element of a Banach space $X = (X_1,\ldots,X_n)$. In general, $n$ is the total number of variables defining the database. While it seems the most intuitive to think that each $X_i$ corresponds to some cell of a table in the database, it is possible that $X_i$ in turn denotes a \emph{subspace} of variables, i.e. a row or a column.

Let $q: X \to Y$ be a query. Let the noised query $\noised{q}$ be differentially private variant of $q$. Let $X = (X_1,\ldots,X_n)$, and let $X_i$ be the component that the attacker wants to guess (e.g. a subset of table columns). Let us constrain Def.~\ref{def:gacont} and take $g(x) = x_k$ and $e(x) = (x_1,\ldots,x_{k-1},x_{k+1},\ldots,x_n)$. Defining guessing advantage w.r.t. component, will be related to Def~\ref{def:dpcomp} of DP w.r.t. component.

\begin{definition}\label{def:ga}[Guessing advantage w.r.t. component]
Let $x = (x_1,\ldots,x_n) \in (X_1,\ldots,X_n)$ be the data instance, $\noised{q}(x)$ the observation. We say that advantage of guessing $x_k$ with precision $r_k$ is at most $\delta$ if
\[\abs{\pr{\mathcal{B}(x_k,r_k)\ |\ \noised{q}(X) = \noised{q}(x), \bigwedge_{i \neq k} (X_i = x_i)} - \pr{\mathcal{B}(x_k,r_k)\ |\ \bigwedge_{i \neq k} (X_i = x_i)}} \leq \delta\enspace,\]
where $\mathcal{B}(x_k,r_k) = \set{x'_k | d(x_k,x'_k) \leq r}$.
\end{definition}

Formally, there is a collection $\mathcal{S} = \set{S_k\ |\ k \in [|\mathcal{S}|]}$ of sets of sensitive attributes that the attacker wants to guess. Each $S_k \in \mathcal{S}$ is of the form $S_k = \set{(x_{k_1}, r_{k_1}),\ldots,(x_{k_{n_k}},r_{k_{n_k}})}$, denoting the set of variables $x_i$ that are not allowed to be disclosed to the attacker with precision at least $r_i$. All the attributes of $S_k$ define one disclosure, i.e., the attacker needs to guess \emph{all} $x_i$ in order to win.

\subsection{The Univariate Case}\label{sec:attadv:univar}

First, we show how to compute the $\epsilon$ for a single set $S_k \in \mathcal{S}$, where $S_k = \set{(t,r)}$ is univariate. The set of ``sufficiently correct guesses'' is defined as $X' = \set{x'\ |\ d(t,x') \leq r}$. Let $R$ be such that $\pr{x\ |\ \forall x' \in X':\ d(x,x') \leq R} \approx 1$.

As discussed in Sec.~\ref{sec:gadp}, to achieve $\prpost{X'} \leq \prpre{X'} + \delta$, we can come up with $a \leq R$ defining $\hat{X'_a} := \set{x\ |\ \forall x' \in X':\ d(x,x') \leq a} \setminus X'$. We get $\hat{X'_a} := \set{x\ |\ r < d(x,x') \leq a}$. It is more complicated with defining $a$ for $\prpost{X'} \geq \prpre{X'} - \delta$, as the elements inside $X\setminus X'$ can already be $R$ units apart from each other.

It can be reasonable to adjust Def.~\ref{def:ga} and deliberately reduce the search space of the attacker, taking $X := X' \cup \hat{X}_a$. Formally, we increase the knowledge of the attacker and add to $e(x)$ the information that the input is located inside $\hat{X}_a$. Note that we do not rescale the prior $\prpre{X'}$, but just ``squeeze'' the remaining probability weight so that $\prpre{\hat{X}_a} := \prpre{X \setminus X'}$. As the result, we now only claim that the prior and the posterior distributions are different for $x$ such that $d(x,t) \leq a$. This allows to consider distances up to $a$. The smaller $a$ is, the closer is $\prpost{X'}$ to $1$.

It remains to show how the values $\prpre{X'}$ and $\prpre{\hat{X}_a}$ are actually computed. Let $g$ be a function such that $g(z) := \prpre{x\ |\ d(x,t) \leq z}$. The input $t$ is implicitly included into description of $g$. We have $\prpre{X'} = g(r)$ and  $\prpre{\hat{X'_a}} = g(a) - g(r)$. Let us investigate some special cases of $g(z)$. We also see whether we are able to make any reasonable analysis if prior probabilities are not known in advance, and we look more precisely at discrete datasets.

\paragraph{Uniform distribution.}

Given a radius $r$ and an upper bound $R$ on $d(x,x')$, it is easy to compute $\prpre{X'} = \frac{2r}{R}$. For $a < R$, we can as well take $g(a) = \frac{2a}{R}$. This is true of we can perfectly fit $\hat{X'_a}$ into $X$. The problem is that the input $t$ may be located not in the ``center'' of the space $X$, but somewhere in the corner, as shown in Figure~\ref{fig:corner}. It may happen that we cannot come up with a set $\hat{X'_a}$ with probability weight $\frac{2(a - r)}{R}$ that would satisfy $d(x,x') \leq a$ for all $x \in \hat{X'_a}$ and $x' \in X'$, and need to include more distant points into $\hat{X'_a}$. We may have distance up to $d(x,x') \leq 2a$ for $x' \in X'$, $x \in \hat{X'_a}$, which means that $2$ will go to the exponent of $e$. To avoid the change in exponent, we can just take $a' := 2a$, getting $\prpre{X'} = \frac{2r}{R}$ and $g(a') = \frac{a'}{R}$.

\begin{wrapfigure}[11]{R}{.33\textwidth}
\center
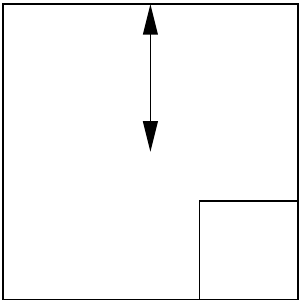
\caption{Bad location of $X'$}\label{fig:corner}
\end{wrapfigure}

\paragraph{Distributions with well-defined cumulative distribution function (CDF).}

We can estimate noise for normally distributed inputs, and in general for data with well-defined CDF $F(x)$. Namely, since we define $g(z) = \set{x | d(x,t) \leq z}$, then by definition of CDF we get $X' = F(x + z) - F(x - z)$.

If we are not given an upper bound $R$ on $X$, then we can still derive some $R$ based on the distribution. If the distribution has bell shape, it does not make sense to consider elements that are too far from the center. For example, for normal distribution $N(\mu,\sigma^2)$ we can take $R = \mu + 3\cdot\sqrt{2}\sigma$, which covers $\erf{3} \approx 0.9999779$ of the input space and is more than enough.

\paragraph{The worst prior probability.}

Even if the input distribution is unknown, we can still define $\prpre{X'}$ in such a way that the maximum noise will be needed, so that our analysis would be valid for any possible prior. This value is non-trivial, e.g. although $p = 1.0$ increases posterior probability the most, the advantage would be $0$. We want to find $p$ that maximizes
\[\epsilon =: f(p) = \frac{-\ln\left(\frac{p}{1-p} \cdot (\frac{1}{\delta + p} - 1)\right)}{R}\enspace.\]

We can do it by using common calculus. We have
\[\frac{df}{dp}(p) = \left(\frac{1}{(p+\delta) - (p+\delta)^2} - \frac{1}{(1 - p) - (1 - p)^2}\right)\cdot\frac{1}{R}\enspace,\]
\noindent and we have $\frac{df}{dp}(p) = 0$ for $p = \frac{1 - \delta}{2}$.

Is it reasonable to try out $a < R$ if the input distribution is unknown? In general, if we have no additional knowledge, we only know that $g(R) = 1$. Alternatively, if we do know that $g(a) = q$ for some $q$, then we can take $q - p$ instead of $1 - p$. For example, even if we do not know details of the input distribution, if it is known that most of the data stays within some $R' < R$, and $\prpre{x\ |\ d(t,x) \leq R'} = q$, it can be useful to give take $a := R'$.

The question is how to choose the worst-case $p$. Substituting $1 - p$ with $g(a) - p$, we get

\[\frac{df}{dp}(p) = \left(\frac{1}{(p+\delta) - (p+\delta)^2} - \frac{1}{(q - p) - (q - p)^2}\right)\cdot\frac{1}{R}\enspace,\]
\noindent and we have $\frac{df}{dp}(p) = 0$ for $p = \frac{\delta(1-\delta) + q(1-q)}{2(\delta + q - 1)}$. This is in general different from the $q=1$ case. However, note that in practice it is fine to consider only the case $q = 1$. If we have found $p$ that minimizes $\epsilon$ in the case $a = R$, this $\epsilon$ is in any case sufficient. The aim of taking $a < R$ is only to check whether it is possible to come up with a better $\epsilon$. Hence, we can always take $p = \frac{1-\delta}{2}$.

\paragraph{A discrete set of prior probabilities.}

Suppose that the possible value of $x$ may come from the set of values $\set{x_1,\ldots,x_n}$ with probabilities $\set{p_1,\ldots,p_n}$ respectively ($\sum_{i=1}^{n}p_i = 1$). Here $x$ is not necessarily a single attribute, and the set $\set{p_1,\ldots,p_n}$ may be computed from all possible combinations of several attributes, which in this case are allowed to be correlated.

If we have access to the actual data, then we may just look at the actual $t$ and take $p = p_k$ such that $t = x_k$, which guarantees that the analysis will be linear w.r.t, the number of table rows. Without having access to the actual data (which may happen if global sensitivity is used to enforce DP), we would need to make a theoretical estimate. We need to take the worst $p_k$ to ensure that each row is protected. If $n \leq \infty$, we can do it for each $p_k$ one by one, but it is not computationally reasonable. Instead, we may start directly from $p_k$ that leaks the most. In the previous paragraph we have shown that the worst case is $p' = \frac{1}{1 + e^{-\epsilon/2}}$. This value itself may be missing from $\set{p_1,\ldots,p_n}$. However, since the function $f(p)$ has exactly one local maximum at $p'$ and is monotone at $(-\infty,p']$ and $[p',\infty)$, the values $p_l$ and $p_r$ that are closest to $p'$ from left and right respectively are the worst cases. It suffices to compute $\epsilon$ only for these two values, and take the smallest one.

\subsection{The Multivariate Case: AND-events}\label{sec:andevent}

Let us compute the $\epsilon$ for a single set $S_k \in \mathcal{S}$, where $S_k = \set{(t_1,r_1),\ldots,(t_n,r_n)}$.  Assume that there is an upper bound $R_i$ on each dimension. The univariate case can be generalized to multivariate, treating $X = (X_1,\ldots,X_n)$ as a single vector variable, taking $\vec{t} = (t_1,\ldots,t_n)$, $\vec{r} = (r_1,\ldots,r_n)$, $\vec{R} = (R_1,\ldots,R_n)$. For this, we need to clarify how the distance in $X$ should be defined. If it is an $\ell_p$ norm of underlying dimensions, we will have $d(t,x) = \norm{t_1 - x_1,\ldots,t_n - x_n}_{p}$. According to our intuition, the attacker wins if he guesses correctly a point from the rectangle $[t_1 - r_1,t_1 + r_1] \times \cdots \times [t_n-r_n,t_n + r_n]$. Hence, in general we take $p = \infty$. In some cases, we may as well be interested in other value of $p$, e.g. if $t_1$ and $t_2$ are coordinates of geographic location, then $p = 2$ corresponds more to our intuition.

In Sec.~\ref{sec:gadp}, the numbers $r$ and $a$ are single numbers, and not vectors. In general, we should take $R = \norm{\vec{R}}_{p}$, $r = \norm{\vec{r}}_{p}$ and $a = \norm{\vec{a}}_{p}$. However, it is not clear how to come up with $\vec{a}$ in the first place. First of all, let us \emph{scale} all dimensions to get $r := r_i = r_j$ for all $i,j$, which allows us to define a single $r$. E.g. take $r = 1$ and scale all $R_i$ accordingly. We can now optimize $a$ based on this $r$. Still, as the scaled $R_i$ can be different, the same value $a$ might not fit into every dimension. Optimizing a multivariate function by trying out all $a_i \in \set{r + \frac{k}{N} \cdot (R_i - r)\ |\ k \in \set{1,\ldots,N}}$ is too expensive. Hence, we first sample $r < a \leq R$, which we will use for distance. We will need particular $a_i$ to compute the probability weights $g(a_i)$, and we take $a_i = \min(a,R_i)$.%Indeed, the obtained $a_1,\ldots,a_n$ may be non-optimal.

In Sec.~\ref{sec:attadv:univar}, we have shown how to compute $g(z_i)$ for one-dimensional $z_i$. Assuming that the variables are independent, we can compute $g(z) = \prod_{i} g(z_i)$ for an AND-event. If the variables are not independent, then there will be some special way of defining $g(z)$, which depends on the distribution. The relations between variables can make only certain combinations of $(a_1,\ldots,a_n)$ possible. However, since $a_i$ are only used in $g(a_i)$ anyway, it suffices to work with $a$ only. Thus, a subspace of correlated variables would be no different from the univariate case, and the probability distribution over that subspace should be given in advance.

The discussion above can be summarized into Theorem~\ref{thm:andevent}.
\begin{theorem}\label{thm:andevent}
Let $X  = (X_1,\ldots,X_n)$, where $X_i$ are pairwise independent, and $R_i = \max_{x_i,x'_i \in X_i} d(x_i,x'_i)$. Let $(t_1,\ldots,t_n)$ be the actual data, and let $X' = \set{\vec{x} \in X\ |\ d(t_i,x_i) \leq r_i}$. Let $g_i(z) := \prpre{x_i\ |\ d(x_i,t_i) \leq z}$. Let $\noised{q}$ be an $\epsilon$-DP mechanism w.r.t. norm $\norm{x_1,\ldots,x_n}_{\infty}$. To bound the guessing advantage of $X'$ by $\delta$, we need to take

\[\epsilon \leq \frac{-\ln\left(\frac{\prod_{i=1}^{n} g_i(r_i)}{\prod_{i=1}^{n} g_i(a_i) - \prod_{i=1}^{n} g_i(r_i)} \cdot (\frac{1}{\delta + \prod_{i=1}^{n} g_i(r_i)} - 1)\right)}{a}\enspace.\]

\noindent where $a_i$ is a freely chosen value satisfying $r_i < a_i \leq R_i$, and $a = \norm{a_1,\ldots,a_n}_{\infty}$.
\end{theorem}

\subsection{The Multivariate Case: OR-events}\label{sec:orevent}

There are two ways to define guessing advantage for entire $\mathcal{S} = \set{S_k\ |\ k \in [|\mathcal{S}|]}$:
\begin{enumerate}
\item\label{item:v1} Each $S_k \in \mathcal{S}$ is protected independently. E.g. if we have two uniformly distributed binary variables, then advantage $\delta = 10\%$ means that the probability of guessing increases from $50\%$ to $60\%$ for \emph{any} of these variables. For each $S_k \in \mathcal{S}$, we need to satisfy Def.~\ref{def:ga} with $X_k = \set{x\ |\ \bigwedge_j x_{i_j} \leq r_{i_j}}$. The suitable $\epsilon$ for differential privacy noise can be found as described in Sec.~\ref{sec:attadv:univar} (for univariate case) and Sec.~\ref{sec:andevent} (for multivariate case). We take the maximum of obtained noise magnitudes, which is sufficient to protect \emph{each} sub{\sensdata} of $\mathcal{S}$.
\item\label{item:v2} Alternatively, we can estimate the probability that the attacker guesses \emph{at least one} {\sensdata}. This poses a different question for an OR condition, and the prior probability gets completely different meaning. E.g. if we have two uniformly distributed binary variables, we will have $75\%$ of correct and $25\%$ of incorrect guesses, and $\delta = 10\%$ means that the probability of guessing \emph{at least one} of these two increases from $75\%$ to $85\%$. This approach assumes $X_j = \set{x | \bigvee_k \bigwedge_j x_{i_j} \leq r_{i_j}}$ in Def.~\ref{def:ga}.
\end{enumerate}

Variant~(\ref{item:v1}) is based purely on results of Sec.~\ref{sec:attadv:univar} and Sec.~\ref{sec:andevent}. In this section, we show how to deal with variant~(\ref{item:v2}).

It is easier to estimate approximation of an AND of conditions, since the set $X'$ is bounded by $r_i$ at the coordinate $X'_i$, and $d(x,x') \leq r = \norm{r_1,\ldots,r_n}_p$ for all $x,x' \in X$. However, if the attacker wants to guess an OR of conditions, then he may guess $x_i$ correctly even if some other variable $x_j$ is as far from the actual value $t_j$ as possible, and we may only claim that $d(x,x') \leq R = \norm{R_1,\ldots,R_n}_{\infty}$. A comparison of AND and OR sets is depicted in Figure~\ref{fig:andor2D}.

If we take $a_i = R_i$ for all $i$, then the only difference of an OR-event from an AND-event is the construction of $\prpre{X'}$. We have $\prpre{X'} = \prod_{i=1}^n g(a_i) - \prod_{i=1}^n g(a_i - r_i)$.

Let us see what happens if we want to use $a_i < R_i$. Now the problem is that we may have $d(\vec{x},\vec{x'}) > a$ for $\vec{x},\vec{x'} \in X'$, as the distances inside $X'$ are not bounded by $r$. We need to approach it differently.

Let $g_{\vec{t'}}(z) = \pr{\set{x\ |\ d(x_1,t'_1) \leq z \vee \cdots \vee d(x_n,t'_n) \leq z}}$ for $\vec{t'} \in X$. Let $\vec{t} = (t_1,\ldots,t_n)$ be the actual datapoint. %We assume that either $g_{\vec{t'}}(z)$ does not depend on $\vec{t'}$ at all (i.e. the distribution is uniform), or perform the analysis for $\vec{t'} = \mathsf{argmax}_{\vec{x'} \in X'}(\prod_{i=1}^n g(a_i) - \prod_{i=1}^n g(a_i - r_i))$. This property is very important, as otherwise a point $\vec{t'}$ arbitrary far from $\vec{t}$ can contribute to $\prpre{X'}$ more than the neighbourhood of $\vec{t}$. Let $g = g_{\vec{t'}}$.

We split both $X'$ and $\hat{X'}$ into the same number of blocks. The idea is that each block $X'_k$ of $X'$ would have a sufficiently close unique neighbour block $\hat{X'}_k$ in $\hat{X'}$. Let $\vec{a} = (a_1,\ldots,a_n)$, and $a = 2\norm{a_1,\ldots,a_n}_{\infty}$. First, we define
\[X'_0 := \set{\vec{x} | \bigwedge_{i=1}^n x_i \in \mathcal{B}(t_i,a_i) \wedge \left(\bigvee_{i=1}^{n} x_i \in \mathcal{B}(t_i,r_i)\right)}\]
and
\[\hat{X'_0} := \set{\vec{x} |\bigwedge_{i=1}^n x_i \in \mathcal{B}(t_i,a_i)} \setminus X'_0\enspace.\]
We also define
\[X^{\ell}_k = \set{\vec{x} | x_k \in \mathcal{B}(t_i,r_i) \wedge \bigwedge_{i=1}^n x_i \in \mathcal{B}(t_i + 2\ell_i \cdot a_i,a_i)}\]
and
\[\hat{X^{\ell}_k} = \set{\vec{x} | \bigwedge_{i=1}^n x_i \in \mathcal{B}(t_i + 2\ell_i \cdot a_i,a_i)} \setminus \hat{X'_0}\]
for $\ell \in \mathcal{L} \subseteq \ZZ^{n}$ and $k \in [n]$, where $\mathcal{L}$ depends on the size of the total space $X$. All blocks are pairwise disjoint, and $X'_0 \cup \bigcup_{\ell,k}X'^{\ell}_k = X'$. An example of splitting into blocks is given in Figure~\ref{fig:split}.

\begin{figure}
\begin{subfigure}{0.4\textwidth}
\center
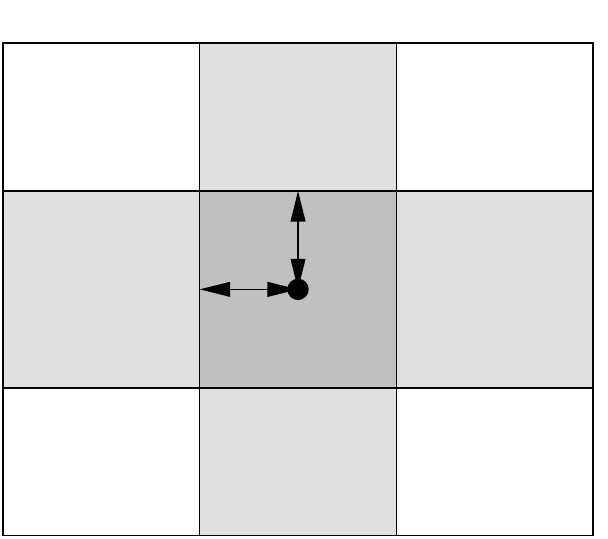
\caption{The area of correctly guessing AND and OR of two approximations}\label{fig:andor2D}
\end{subfigure}
\hfill
\begin{subfigure}{0.5\textwidth}
\center
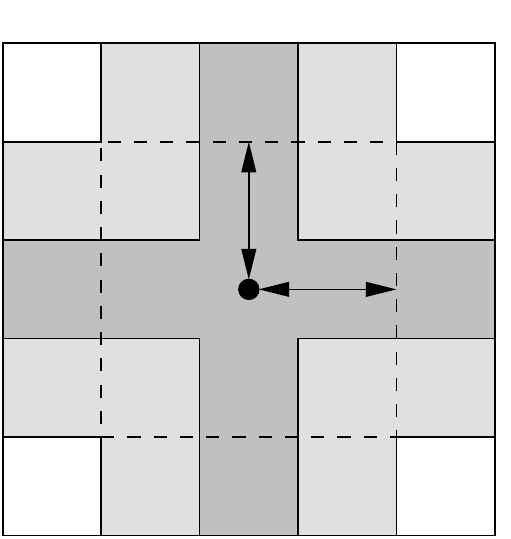
\caption{Splitting $X'$ and $\hat{X'}$ into easily comparable blocks}\label{fig:split}
\end{subfigure}
\end{figure}

The constructed pairs enjoy the following important properties:
\begin{enumerate}
\item\label{pr:1} For all $\vec{x'} \in X'_0$, $\vec{x} \in \hat{X'_0}$ (and also $\vec{x'} \in X^{\ell}_k$, $\vec{x} \in \hat{X^{\ell}_k}$), we have $d(\vec{x},\vec{x'}) \leq a$;
\item\label{pr:2} We have $\frac{\prpost{X'_0}}{\prpost{\hat{X'_0}}} = \frac{g_{\vec{t}}(a_1,\ldots,a_n) - g_{\vec{t}}((a_1-1,\ldots,a_n-1)))}{g_{\vec{t}}((a_1-1,\ldots,a_n-1))} =: \alpha_0$;
\item\label{pr:3} For all $k \in [n]$, $\ell \in \mathcal{L}$, we have $\frac{\prpost{X^{\ell}_k}}{\prpost{\hat{X^{\ell}}_k}}= \frac{g_{\vec{t'}}(a_1,\ldots,r_k,\ldots,a_n)}{g_{\vec{t'}}(a_1,\ldots,a_k,\ldots,a_n) - g_{\vec{t'}}(a_1,\ldots,r_k,\ldots,a_n)}$, where $\vec{t'} = (t_1 + 2\ell_1\cdot a_1,\ldots,t_k,\ldots,t_n + 2\ell_n \cdot a_n)$. Assuming that the variables are independent, this equals $\frac{g_{t_k}(r_k)}{g_{t_k}(a_k) - g_{t_k}(r_k)} =: \alpha_k$.
\end{enumerate}

We will now analyze $\frac{\prpost{X'_0}}{\prpost{\hat{X'}_0}}$. By property~(\ref{pr:1}), $d(\vec{x},\vec{x'}) \leq a$, so $\epsilon$-DP gives $\fy(y\ | \vec{x'}) \leq e^{a\epsilon}\cdot\fy(y\ | \vec{x})$. We have
\begin{eqnarray*}
\frac{\prpost{X'_0}}{\prpost{\hat{X'}_0}} & = & \frac{\int_{X'_0}\fy(y\ | \vec{x'})\cdot\fxpre(\vec{x'})\ d\vec{x'}}{\int_{\hat{X'_0}}\fy(y\ | \vec{x})\cdot\fxpre(\vec{x})\ d\vec{x}}\\
& \leq & e^{a\epsilon}\cdot\frac{\int_{X'_0}\fxpre(\vec{x'})\ d\vec{x'}}{\int_{\hat{X'_0}}\fxpre(\vec{x})\ d\vec{x}} = e^{a\epsilon}\cdot\frac{\prpre{X'_0}}{\prpre{\hat{X'}_0}}= e^{a\epsilon}\cdot\alpha_0\enspace.
\end{eqnarray*}

Similarly, for $k \in [n]$, $\ell \in \mathcal{L}$, we have
\begin{eqnarray*}
\frac{\prpost{X^{\ell}_k}}{\prpost{\hat{X^{\ell}}_k}} & \leq & e^{a\epsilon}\cdot\frac{\prpre{X^\ell_k}}{\prpre{\hat{X^{\ell}}_k}}= e^{a\epsilon}\cdot\alpha_k\enspace.
\end{eqnarray*}

Let $\alpha = \max_{k \in \set{0,1,\ldots,n}} \alpha_k$. We get

\begin{eqnarray*}
\frac{\prpost{X'}}{\prpost{\hat{X'}}} & = & \frac{\prpost{X'_0} + \sum_{k \in [n],\ell\in \mathcal{L}}\prpost{X^{\ell}_k}}{\prpost{\hat{X'_0}} + \sum_{k \in [n],\ell \in \mathcal{L}}\prpost{\hat{X^{\ell}_k}}}\\
& \leq & \frac{e^{a\epsilon}\cdot\alpha \cdot \prpost{\hat{X'}_0} + \sum_{k \in [n],\ell\in \mathcal{L}}e^{a\epsilon}\cdot\alpha\cdot\prpost{\hat{X^{\ell}}_k}}{\prpost{\hat{X'_0}} + \sum_{k \in [n],\ell \in \mathcal{L}}\prpost{\hat{X^{\ell}_k}}}\\
& = & e^{a\epsilon}\cdot\alpha \enspace.
\end{eqnarray*}

The final result depends on $l = \mathsf{argmax}_{k \in [n]} \alpha_k$.

\begin{enumerate}
\item Let $l = 0$. By property~(\ref{pr:2}), we have $\frac{\prpre{X'_0}}{\prpre{\hat{X'_0}}} = \frac{g_{\vec{t}}(a_1,\ldots,a_n) - g_{\vec{t}}((a_1-1,\ldots,a_n-1)))}{g_{\vec{t}}((a_1-1,\ldots,a_n-1))}$. Applying Equation~\ref{eq:main}, we get some $\epsilon_0$.
%\[\prpost{X'} \leq \frac{1}{1 + e^{-a\epsilon}\cdot\frac{g_{\vec{t}}((a_1-1,\ldots,a_n-1))}{g_{\vec{t}}(a_1,\ldots,a_n) - g_{\vec{t}}((a_1-1,\ldots,a_n-1)))}}\enspace.\]
\item Let $l = k \neq 0$. By property~(\ref{pr:3}), we have $\frac{\prpre{X^\ell_k}}{\prpre{\hat{X^{\ell}_k}}} = \frac{g_{t_k}(r_k)}{g_{t_k}(a_k) - g_{t_k}(r_k)}$ for any $\ell \in \mathcal{L}$. Applying Equation~\ref{eq:main}, we get some $\epsilon_k$.
\end{enumerate}

We take $\epsilon = \min_{k \in \set{0,1,\ldots,n}}(\epsilon_k)$. This shows that, in addition to protecting each dimension separately, we also need to take care of the $a_1 \times \cdots \times a_n$ block surrounding the actual input $\vec{t}$. We emphasize that it holds only for independent variables, as otherwise we cannot treat different $\epsilon_k$ values independently. %A counterexample is given in Figure~\ref{fig:dependentvars}, where $\prpre{X'_0} = 0$, which already implies that no noise is needed to hide $\prpre{X'_0}$ at all. However, we would actually need to take care of the blocks that go beyond $a$ (colored dark gray in the figure), as the attacker's guess would be still correct in these cases.

%\begin{figure}
%\center
%\input{dependent.pdf_t}
%\caption{$X'$ with dependent variables (white areas have probability weight $0$)}\label{fig:dependentvars}
%\end{figure}

The discussion above can be summarized into Theorem~\ref{thm:orevent}.
\begin{theorem}\label{thm:orevent}
Let $X  = (X_1,\ldots,X_n)$, where $X_i$ are pairwise independent, and $R_i = \max_{x_i,x'_i \in X_i} d(x_i,x'_i)$. Let $(t_1,\ldots,t_n)$ be the actual data, and let $X' = \set{\vec{x} \in X\ |\ \bigvee_{i=1}^{n}d(t_i,x_i) \leq r_i}$. Let $g_i(z) := \max_{t_i \in X'_i}\prpre{x_i\ |\ d(x_i,t_i) \leq z}$. Let $\noised{q}$ be an $\epsilon$-DP mechanism w.r.t. norm $\norm{x_1,\ldots,x_n}_{\infty}$. To bound the guessing advantage of $X'$ by $\delta$, we need to take $\epsilon = \min_{k \in \set{0,1,\ldots,n}}(\epsilon_k)$, where
\begin{eqnarray*}
\epsilon_0 &=& \frac{-\ln\left(\frac{\prod_{i=1}^{n} g(a_i) - \prod_{i=1}^{n} g(a_i-1)}{\prod_{i=1}^{n} g(a_i-1)} \cdot (\frac{1}{\delta + \prod_{i=1}^{n} g(a_i) - \prod_{i=1}^{n} g(a_i-1)} - 1)\right)}{a}\enspace;\\
\epsilon_k &=& \frac{-\ln\left(\frac{g_{t_k}(r_k)}{g_{t_k}(a_k) - g_{t_k}(r_k)} \cdot (\frac{1}{\delta + g_{t_k}(r_k)} - 1)\right)}{a}\enspace;\\
\end{eqnarray*}
\noindent where $a_i$ is a freely chosen value satisfying $r_i \leq a_i \leq R_i$, and $a = \norm{a_1,\ldots,a_n}_{\infty}$.
\end{theorem}

\section{Queries with Multiple Outputs}\label{sec:multout}

The results of Sec~\ref{sec:andevent} and Sec.~\ref{sec:orevent} do not depend on the definition of the set of outputs $q(X)$. They only tell which $\epsilon$ is needed to satisfy the desired guessing advantage, and our task is to come up with a mechanism achieving this $\epsilon$. If there are multiple outputs, we can use the standard notions of parallel and sequential composition.

\begin{theorem}[Sequential composition~\cite{DBLP:journals/corr/OhV13}]\label{thm:dpcomp}
Let $q(x) = (q_1(x) \in \RR,\ldots,q_m(x) \in \RR)$. Let DP mechanism be such that, for all $i \in [m]$, we have $\pr{\noised{q_i}(x) \in Y} \leq e^{\epsilon_i}\cdot\pr{\noised{q_i}(x') \in Y}$ for all subsets $Y \subseteq \RR$. Then,
\[
\pr{\noised{q}(x) \in Y} \leq e^{\sum_{i=1}^{m}\epsilon_i}\cdot\pr{\noised{q}(x') \in Y}
\]
for all subsets $Y \subseteq \RR^m$.
\end{theorem}

Parallel composition of DP says that we can take $\max(\epsilon_i)$ instead of sum if the variables used by different queries are independent. Roughly speaking, the independence of inputs ensures that the condition $d(t,t') = 1$ affects at most one of the outputs. We can generalize this result using the notion of $\ell_p$-norms.

\begin{theorem}\label{thm:dpcomp2}
Let $q(x_1,\ldots,x_n) = (q_1(x_1) \in \RR,\ldots,q_m(x_m) \in \RR)$. Let the distance $d(x,x')$ for $x,x' \in \RR^m$ be defined as an $\ell_p$-norm. Let a privacy mechanism $\noised{q_i}$ be $\epsilon_i$-differentially-private. Then, the mechanism $\noised{q} = (\noised{q_1},\ldots,\noised{q_m})$ is $\ell_q(\epsilon_1,\ldots,\epsilon_m)$-differentially-private, where $\ell_q$ is the dual norm of $\ell_p$.
\end{theorem}
\begin{proof}
Let $d(x,x') \leq 1$, and let $Y = (Y_1,\ldots,Y_n) \subseteq \RR^m$ be arbitrary. We have $x = (x_1,\ldots,x_n)$ and $x' = (x'_1,\ldots,x'_n)$. For an $\ell_p$-norm, denoting $d_i := d(x_i,x'_i)$ for all $i \in [m]$, we have $d(x,x') = \norm{d_1,\ldots,d_m}_p$. For each $i$, we have $\pr{\noised{q_i}(x_i) \in Y_i} \leq e^{d_i\epsilon_i}\cdot\pr{\noised{q_i}(x') \in Y_i}$. By Theorem~\ref{thm:dpcomp}, $\pr{\noised{q}(x) \in Y} \leq e^{\sum_{i=1}^{m}d_i \epsilon_i}\cdot\pr{\noised{q}(x') \in Y}$, so the mechanism $\noised{q}$ is $\sum_{i=1}^{m}d_i \epsilon_i$-differentially-private. By definition of the dual norm, $\norm{\epsilon_1,\ldots,\epsilon_m}_q = \sup\set{\sum_{i=1}^m \epsilon_i z_i\ |\ \norm{z}_p \leq 1}$. Since $(d_1,\ldots,d_m)$ is one candidate for $z$ in this expression, we have $\sum_{i=1}^{m}d_i \epsilon_i \leq \norm{\epsilon_1,\ldots,\epsilon_m}_q$.\hfill$\square$
\end{proof}

The most intuitive instantiations for Theorem~\ref{thm:dpcomp2} are the cases of $\ell_1$ and $\ell_{\infty}$ norms, where $d(x,x')$ is an integer value, such as the number of rows. These will give us parallel and sequential compositions. In the case of $\ell_1$-norm, $d(x,x') = 1$ means that exactly one of the inputs $x_i$ will change by $1$, so it is sufficient to take $\max_i \epsilon_i$ to protect from one change. In the case of $\ell_{\infty}$-norm, $d(x,x') = 1$ means that each input $x_i$ may change by $1$, so we need to protect all the outputs at once, getting $\sum_i \epsilon_i$.

 We can use Theorem~\ref{thm:dpcomp2} to find appropriate partitionings to $\epsilon_i$, which are not unique, and the partitioning may be optimized in such a way that the error would be minimal, i.e. reserve larger $\epsilon_i$ for more sensitive $q_i$. Finding the optimal partitioning is a possible direction for future research. If the queries are similar (e.g. groups of a \texttt{GROUP BY} query), then the best option is to assign the same value to all $\epsilon_i$.

\section{Examples}

\subsection{Discrete data}\label{sec:example:cats}

Assume that we have a data table \texttt{cat}. Each cat can be male or female, and has one of the five main colors. We want to count all female tabby cats.
\[\text{\texttt{SELECT\ COUNT(*)\ FROM\ cat\ WHERE\ gender=`F'\ AND\ color=`tabby';}}\]
Define the distance between two databases as $1$ iff the gender or the color are different. Such distance is subsumed standard DP metric, as every tabby female is viewed as an included record, and every other combination as excluded ,and sensitivity of a COUNT query w.r.t. this metric is $1$, similarly to sensitivity of a COUNT query w.r.t. standard DP metric. Hence, the query could be executed with obtained $\epsilon$ in PINQ system.

We need to define some prior or the distributions of genders and colors. Let gender be distributed uniformly, and let colors be distributed as
\[\set{red : 0.2, white : 0.1, tabby : 0.25, black : 0.4,tortoise : 0.05}\enspace.\]

\paragraph{AND-event.}
First, assume that the attacker wins if he guesses both gender \emph{and} color. Combining the two sensitive attributes together, we get the probabilities of AND-events
\begin{align*}
\{& (red \wedge M) : 0.1, (white \wedge M) : 0.05, (tabby \wedge M) : 0.125, (black \wedge M) : 0.2,\\
  & (red \wedge F) : 0.1, (white \wedge F) : 0.05, (tabby \wedge F) : 0.125, (black \wedge F) : 0.2,\\
  & (tortoise \wedge M) : 0.025, (tortoise \wedge F) : 0.025\}\enspace.
\end{align*}

We want to bound guessing advantage by $\delta = 0.1$. As shown in Sec.~\ref{sec:attadv:univar}, the worst-case prior would be $p = \frac{1 - 0.1}{2} = 0.45$. The closest to the worst case are $(black \wedge M)$ and $(black \wedge F)$ with $p = 0.2$. We have $p \leq 1 - p$, so it suffices to take
\[\epsilon \leq \frac{-\ln\left(\frac{0.2}{0.8} \cdot (\frac{1}{0.1 + 0.2} - 1)\right)}{1} = -\ln\left(\frac{0.2}{0.8} \cdot\frac{0.7}{0.3}\right) \approx 0.539\enspace.\]

\paragraph{OR-event.}
Let us now assume that the attacker wins if he guesses either gender \emph{or} color. This is equivalent to `not guessing both in a wrong way'. We can write the set of possible correct answers as
\begin{align*}
\{& (red \vee M) : 0.6, (white \wedge M) : 0.55, (tabby \wedge M) : 0.625, (black \wedge M) : 0.7,\\
  & (red \wedge F) : 0.6, (white \wedge F) : 0.55, (tabby \wedge F) : 0.625, (black \wedge F) : 0.7,\\
  & (tortoise \wedge M) : 0.525, (tortoise \wedge F) : 0.525\}\enspace.
\end{align*}
The closest to the worst prior are $(tortoise \wedge M)$ and $(tortoise \wedge F)$ with $p = 0.525$. We have $p \leq 1 - p$, so it suffices to take
\[\epsilon\leq\frac{-\ln\left(\frac{0.475}{0.525} \cdot (\frac{1}{0.1 + 0.475} - 1)\right)}{1} = -\ln\left(\frac{0.475}{0.425} \cdot\frac{0.425}{0.575}\right) \approx 0.402 \enspace.\]

As sensitivity of COUNT query w.r.t. proposed metric is $1$, using Laplace mechanism (described e.g. in~\cite{DBLP:conf/icalp/Dwork06}), it would be sufficient to add Laplace noise with scaling parameter $\lambda = 1 / 0.539 \approx 1.86$ for the AND-event, and $1 / 0.402 \approx 2.49$ for the OR-event.
\subsection{Continuous data with univariate approximation}\label{sec:example:salary}

Assume that we have a data table \texttt{employee}. Each of the employees has some fixed salary, and we want to find the total employer expenses.
\[\text{\texttt{SELECT\ SUM(salary)\ FROM\ employee;}}\]
Assume that the attacker wins if he guesses someone's salary with precision of $r = 100$ currency units. For prior, let us assume that the salary of each employee is distributed normally according to $\mathcal{N}(\mu = 2000,\sigma^2 = 55556)$ (so $\sigma \approx 235.7$), i.e. $\pr{x \in (1000,3000)} \approx \erf{3} \approx 0.99998$. Using definition of CDF for normal distribution, $\pr{\abs{x - t} \leq r} = \frac{\erf{\frac{t + r - \mu}{\sigma\sqrt{2}}} - \erf{\frac{t - r - \mu}{\sigma\sqrt{2}}}}{2} = \frac{\erf{\frac{t - 1900}{333}} - \erf{\frac{t - 2100}{333}}}{2}$. Since there is no particular $t$ available, we consider the worst case. 

If we want to bound guessing advantage by $\delta = 0.1$, the worst-case prior would be $p = \frac{1 - 0.1}{2} = 0.45$. The prior closest to this value is achieved for $t = \mu = 2000$ for $p \approx 0.42$. Let us now take $a = 200$. We have $q = \erf{a / (\sigma \cdot \sqrt{2})} \approx 0.6$. Hence,

\[\epsilon \leq \frac{-\ln\left(\frac{0.42}{0.6 - 0.42} \cdot (\frac{1}{0.1 + 0.42} - 1)\right)}{200} = \frac{0.767}{200} \approx 0.0038\enspace.\]

While $\epsilon = 0.0038$ seems extremely small, note that such $\epsilon$ would ensure a very strong definition of differential privacy, where neighbour tables are defined as being $a = 200$ units apart. Alternatively, we could rescale the entire space e.g. by $a = 200$, getting $\epsilon = 0.767$. Such $\epsilon$ ensures differential privacy for neighbour databases defined as differing in $1$ unit. Since each currency unit contributes $1$ unit to the sum, the sensitivity of the query w.r.t. latter distance is $1$. Laplace noise with scaling $1 / \epsilon \approx 1.3$ is sufficient to ensure the guessing advantage below $0.1$, and the noise level is comparable to example of Sec.~\ref{sec:example:cats}. We see that dividing $\epsilon$ by $200$ is fine since we are just working with different magnitudes. Also, such a SUM query would have larger output than a COUNT query (for a similar number of rows), so the noise would have less impact than in example of Sec.~\ref{sec:example:cats}.

\subsection{Continuous data with multivariate approximation}\label{sec:example:ships}

Assume that we have a data table \texttt{ship}. Each ship has some geographic location and certain maximum speed. At some moment, the ships start sailing with their maximum speed towards the port located at the point $(0,0)$. We want to know when the first ship arrives at the port. Using the operator \texttt{<@>} for geographical distance, which is essentially $\ell_2$-norm, the query looks like
\[\text{\texttt{SELECT\ MIN((POINT(ship.x,ship.y)<@>(0,0))/ship.speed)\ FROM\ ship;}}\]
Assume that the attacker wins if he guesses some ship's location with precision of $r = 10$ units w.r.t. $\ell_2$-norm. Let the prior be defined as the distribution of \emph{distances} from the point $(0,0)$, which is $\mathcal{N}(\mu = 0,\sigma^2 = 1250)$ (so $\sigma \approx 35.4$), i.e. $\pr{(x,y) \in (-150,150)} \approx \erf{3} \approx 0.99998$. Note that the variables $x$ and $y$ do not need to be independent, as we work with their joint distribution directly.

Let $(t_x,t_y)$ be the actual location of a ship. Differently from Sec.~\ref{sec:example:salary}, we can find only a lower bound
\begin{eqnarray*}
\pr{\norm{(x,y) - (t_x,t_y)}_2 \leq r} & \geq & \pr{(\norm{((x,y) - (0,0)}_2 + \norm{(t_x,t_y) - (0,0)}_2 \leq r)}
\end{eqnarray*}
\begin{eqnarray*}
& = & \frac{\erf{\frac{\norm{(x,y)}_2 + r - \norm{(t_x,t_y)}_2 - \mu}{\sigma\sqrt{2}}} - \erf{\frac{\norm{(x,y)}_2 - r + \norm{(t_x,t_y)}_2 - \mu}{\sigma\sqrt{2}}}}{2}\\
& = & \frac{\erf{\frac{\norm{(x,y)}_2 - \norm{(t_x,t_y)}_2 + 10}{50}} - \erf{\frac{\norm{(x,y)}_2 + \norm{(t_x,t_y)}_2 - 10}{50}}}{2}\enspace.
\end{eqnarray*}
Bounding guessing advantage by $\delta = 0.1$, the worst-case prior would be $p = \frac{1 - 0.1}{2} = 0.45$. For the lower bound considered above, the prior closest to this value is achieved for $(t_x,t_y) = (0.0)$ for $p \approx 0.22$. Note that in this case inequality becomes an equality, so $p$ cannot be larger. Let us now take $a = 20$. We have $q = \erf{a / (\sigma \cdot \sqrt{2})} \approx 0.43$. Hence,

\[\epsilon \leq \frac{-\ln\left(\frac{0.22}{0.43 - 0.22} \cdot (\frac{1}{0.1 + 0.22} - 1)\right)}{10} = \frac{1.25}{20} \approx 0.063\enspace.\]

Similarly to Sec.~\ref{sec:example:salary}, such $\epsilon$ ensures differential privacy where neighbour tables are defined as being $a = 20$ units apart. Scaling the space by by $a = 20$, we get $\epsilon = 1.25$. Let us see how much one distance unit contributes to the final result. If the ship speed is $v$, then changing its location by $1$ will change the arrival time by $\frac{1}{v}$. This will not affect the minimum if that ship was not among the first ones, but in the worst case the query output will also change by $\frac{1}{v}$ as well. Let $v$ be the minimal speed of a ship, e.g. $v = 1$, so that the query sensitivity would be $1$ similarly to the previous examples. Laplace noise with scaling $1 / \epsilon \approx 0.8$ is sufficient to ensure the guessing advantage below $0.1$.

The noise level seems less than in Sec.~\ref{sec:example:salary}. But is it much or few? This depends on the magnitude of query outputs. While the same noise level could be small for SUM and COUNT queries on large tables, it could be worse for a MIN-query where the output itself is a small quantity. To estimate the goodness of noise, we need the actual data, and instead of \emph{absolute} error estimate the \emph{relative} error as done e.g. in~\cite{dersens}.

\section{Conclusion}

We have shown how to convert guessing advantage to epsilon. We see from some examples that the proposed method can give us reasonable noise.

While we have shown how to work with $\epsilon$-DP, there exist alternative definitions. Relations between $(\epsilon,\delta)$-DP and the guessing advantage are considered in App.~\ref{sec:epsilondelta}. The proposed approach is tightly bound to a particular DP mechanism, and we do not have generic results for $(\epsilon,\delta)$-DP. Extending it to some more general constructions would be an interesting future work.

\bibliographystyle{plain}
\bibliography{references}

%====================================================================================
\appendix

\section{Optimizing the radius of elements that we want to make indistinguishable from the input}\label{app:choosinga}
Let us assume that the inputs are distributed uniformly over a subset of $n$-dimensional Banach space with some $\ell_p$-norm. The hypervolume of a ball of radius $r$ in an $n$-dimensional space is $c_p \cdot r^n$ for a constant $c_p$ that depends on the particular used $\ell_p$-norm, so if the maximum possible radius is $R$, we have $\pr{X'} = \frac{c_p\cdot r^n}{c_p \cdot R^n} = \left(\frac{r}{R}\right)^n$.

A straightforward solution would be to find $a$ that maximizes $\epsilon$ directly. As $\epsilon$ and guessing advantage are correlated, it seems easier to first find $a$ that maximizes the advantage, and then compute the corresponding $\epsilon$.
We get
\begin{eqnarray*}
e^{-a\epsilon} \frac{\prpre{X'_a}}{\prpre{X'}} & = & e^{-a\epsilon} \cdot \frac{\left(\frac{a}{R}\right)^n - \left(\frac{r}{R}\right)^n}{\left(\frac{r}{R}\right)^n}\\
& = & e^{-a\epsilon}\cdot\left(\frac{a^n - r^n}{r^n}\right)\enspace,
\end{eqnarray*}
which allows to find the optimal $a$ for the given $\epsilon$, that maximizes the quantity.

For $n = 1$, this can be done using simple calculus, and it gives $\prpost{X'} \leq \frac{1}{1 + (e^{\epsilon r + 1}\cdot \epsilon r)^{-1}}$ for $a = \frac{1}{\epsilon} + r$. Now the problem is that $a$ in turn depends on $\epsilon$. Recall that our goal is to achieve $\prpost{X'} \leq \prpre{X'} + \delta$ for some predefined advantage $\delta$. Assuming uniform distribution on a subset of $X$, we get $\prpre{X'} = \frac{r}{R}$. Using the equality $\prpost{X'} \leq \frac{1}{1 + (e^{\epsilon r + 1}\cdot \epsilon r)^{-1}}$ we can express $\epsilon$ as follows:
\begin{eqnarray*}
\frac{1}{1 + (e^{\epsilon r + 1}\cdot \epsilon r)^{-1}} & \leq & \frac{r}{R} + \delta \\
1 + (e^{\epsilon r + 1}\cdot \epsilon r)^{-1} & \geq & \frac{1}{\frac{r}{R} + \delta} \\
(e^{\epsilon r + 1}\cdot \epsilon r)^{-1} & \geq & \frac{1}{\frac{r}{R} + \delta} - 1 = \frac{1 - (\frac{r}{R} + \delta)}{\frac{r}{R} + \delta} \\
e^{\epsilon r}\cdot \epsilon r & \leq & \frac{\frac{r}{R} + \delta}{e \cdot (1 - (\frac{r}{R} + \delta))} \enspace.
\end{eqnarray*}
To get $\epsilon r$, we need to find a solution to an equation of the form $x \cdot e^x = y$. The solution is $x = W(y)$, where $W$ is so-called \emph{Lambert function}. We get
\[\epsilon \leq W\left(\frac{\frac{r}{R} + \delta}{e \cdot (1 - (\frac{r}{R} + \delta))}\right) / r\enspace.\]

There is no closed form for $W$, but there exist good approximation methods for computing it. For example, we can use the iterative method of~\cite{lambertIterative}:

\begin{eqnarray*}
w_0(y)     & = & 1 \\
w_{n+1}(y) & = & w_{n-1}(y) - \frac{w_{n-1}(y) \cdot e^{w_n(y)} - y}{(w_n(y) + 1) \cdot e^{w_n(y)}}\enspace.
\end{eqnarray*}

For $n > 1$, it is easier to evaluate directly
\begin{eqnarray}\label{eq:multdim}
\prpost{X'} & = & \frac{\int_{X'} \pr{\noised{q}(x') \in Y} \fxpre(x') dx'}{\int_{X} \pr{\noised{q}(x) \in Y} \fxpre(x) dx} \leq e^{a\epsilon} \int_{X'} \frac{\fxpre(x')}{\int_{x : d(x,x') \leq a} \fxpre(x) dx} dx' \nonumber\\
& = & e^{a\epsilon} \int_{X'} \frac{\fxpre(x')}{\left(\frac{a}{R}\right)^n} dx' = \frac{e^{a\epsilon}}{\left(\frac{a}{R}\right)^n}\left(\frac{r}{R}\right)^n = e^{a\epsilon}\cdot\left(\frac{r}{a}\right)^n \enspace.
\end{eqnarray}
This gives $\prpost{X'} \leq \left(\frac{\epsilon \cdot e \cdot r}{n}\right)^{n}$ at the point $a = \frac{n}{\epsilon}$. Alternatively, we can compute optimal $a$ for each univariate dimension separately, by first computing $\epsilon$ and then turning it into $a = \frac{1}{\epsilon} + r$, and combine as discussed in Sec.~\ref{sec:andevent}.

\subsection{Fixing the Probability that Noise Stays Below a Certain Bound}\label{sec:errorbound}

First of all, we provide a small result that will not only be helpful for estimating $(\epsilon,\delta)$-GA, but is also useful by itself. As described in~\cite{dersens}, we quantify the noise magnitude as $\xi := \frac{c(t)}{b}$, where $c(t)$ is the derivative sensitivity at point $t$ ($t$ is the actual data) and $b$ a parameter that depends on $\epsilon$ and $\beta$, e.g. $b = (\frac{\epsilon}{\gamma} - \beta)$ for Cauchy distribution $\gencauchy{\gamma}$. The additive noise is $\xi \cdot \eta$, where $\eta$ is sampled from Cauchy distribution. Since the value $\xi$ itself does not give user any intuition how large that noise is, we interpreted $\xi$ as the probability that error stays below $\xi$. Knowing that $x \sim \frac{\sqrt{2}}{\pi}\cdot \frac{1}{1 + \abs{x}^4}$, we computed
\[\int_{-1}^{1}\frac{\sqrt{2}}{\pi}\cdot \frac{1}{1 + \abs{x}^4} dx \approx 0.78\enspace,\]
thus fixing the probability that noise stays below $\xi$ to a constant $p = 0.78$. We would like to make our result more flexible and let the user choose $p$. Denoting PDF of noise distribution by $f_{\eta}(x)$, we need to find $a$ such that
\[\int_{-a}^{a} f_{\eta}(x)\ dx = p\enspace.\]
For Laplace distribution, the equation $\int_{-a}^{a}\frac{\epsilon}{2}\cdot e^{-\abs{x}\epsilon}dx = p$ reduces to $1 - e^{-a\epsilon} = p$, and the solution is $a = \frac{-\ln(1 - p)}{\epsilon}$. Unfortunately, there is no nice solution for Cauchy noise. However, since $\int_{-a}^{a} f_{\eta}(x)\ dx = 2 \cdot \int_{0}^{a} f_{\eta}(x)\ dx$, and $f_{\eta}(x)$ is monotone in $[0,\infty)$, we can use window binary search to find $a$ for non-scaled Cauchy distribution. Now, in order to get probability $p$ for $x = \xi\eta$, we just need to take $a\xi$ instead of $a$, as we have
\[p = \int_{-a}^{a} f_{\eta}(x)\ dx = \int_{-a}^{a} f_{\eta}(\xi x)\ d(\xi x) = \int_{-a \xi}^{a \xi} \frac{1}{\xi} \cdot f_{\eta}(\xi x)\ dx\enspace.\]

\section{$(\epsilon,\delta)$-Guessing Advantage}\label{sec:epsilondelta}

So far, we have shown how to convert $\epsilon$ of a DP mechanism to an upper bound guessing advantage. We would like to extend our results to a $(\epsilon,\delta)$-DP mechanism. Unfortunately, it turns out that, for $\delta > 0$, the posterior guessing probability is tightly related to the quantity $f_{Y}(y|x)$, and turns out to be unbounded, i.e. there always exists $y \in Y$ for which the DP mechanism gives no guarantees at all. This happens in the ``tails'' of noise distribution, so although we cannot guarantee privacy for an arbitrary output $y \in Y$, we can compute a bound on guessing advantage $\epsilon'$ that holds with sufficiently high probability $\delta'$, similarly to $(\epsilon,\delta)$-DP definition. In overall, we get that the guessing advantage is bounded by $\epsilon'$ with probability $\delta'$, and it is bounded by $1$ with probability $1 - \delta'$. That is, in average the guessing advantage is bounded by $\epsilon' \cdot \delta' + (1 - \delta') = 1 + (\epsilon' - 1) \cdot \delta'$, giving us a single number, thus allowing to use the original definition of attacker advantage.

\subsection{$(\epsilon,\delta)$-DP to $(\epsilon,\delta)$-GA}\label{sec:dptoga}

Suppose that we have an $(\epsilon,\delta)$ DP mechanism. We show how to convert it to $(\epsilon',\delta')$ guessing advantage. Let $\delta'$ be the probability with which we want the obtained upper bound to hold. We show how to compute an upper bound on guessing advantage $\epsilon'$ that holds with probability $\delta'$. Alternatively, we can fix $\epsilon'$ and find $\delta'$. In both cases, it may happen that there is no $\epsilon' > 0$ or $\delta' > 0$ for the initial choice of $(\epsilon,\delta)$.

Let $t \in X$ be the actual input. Let $X' := \set{x\ |\ d(x,t) < r}$ be the set of inputs for which the attacker's guess is considered correct. Let $\hat{X'} \subseteq X\setminus{X'}$ be such that $\forall x' \in X', x \in \hat{X'}:\ d(x,x') \leq a$ for some $a \in \RR^{+}$. In practice, we take $\hat{X'} := \set{x\ |\ r < d(x,t) \leq a}$. Without loss of generality, we may take $a = 1$ and scale all distances accordingly. We will find $a$ that gives us least noise.

Let $p = \int_{X'}f_X(x)dx$, and $q = \int_{\hat{X'}}f_X(x)dx$. Let us compute the upper bound on posterior probability based on Bayesian inference, similarly to D.1.7. Assuming that the observation $y$ is obtained using $(\epsilon,\delta)$-DP mechanism w.r.t. scaled distance $d'(x,x') := \frac{1}{a} d(x,x')$, we have $f_Y(y|x) \leq e^{\epsilon} \cdot f_Y(y|x') + \delta$. An upper bound on posterior guess is

\begin{eqnarray*}
\prpost{X'} & = & \frac{1}{1 + \frac{\int_{\hat{X'}}f_Y(y|x)f_X(x)dx}{\int_{X'}f_Y(y|x')f_X(x')dx'}}\\
& = & \frac{1}{1 + \frac{\int_{\hat{X'}}\frac{f_Y(y|x)}{f_Y(y|x)e^{\epsilon} + \delta}f_X(x)dx}{\int_{X'}f_X(x')dx'}} = \frac{1}{1 + \frac{\int_{\hat{X'}}\frac{1}{e^{\epsilon} + \frac{\delta}{f_Y(y|x)}}f_X(x)dx}{\int_{X'}f_X(x')dx'}}\\
& \leq & \frac{1}{1 + \frac{1}{e^{\epsilon} + \frac{\delta}{\inf_{x \in \hat{X'}}f_Y(y|x)}}\cdot\frac{q}{p}}\\
%& = & \frac{1}{1 + e^{-\epsilon} \cdot \frac{1}{1 + \frac{\delta}{e^{\epsilon}\inf_{x \in \hat{X'}}f_Y(y|x)}}\cdot\frac{q}{p}}\enspace.
\end{eqnarray*}

We want to bound this quantity from above by $p + \epsilon'$. The problem is with $\underline{\chi}_y := \inf_{x \in \hat{X'}}f_Y(y|x)$. It is possible that $\underline{\chi}_y = 0$ for the given $y \in Y$, which gives a trivial upper bound $\prpost{X'} \leq 1$. We need to introduce the probability $\delta'$ of getting ``sufficiently good'' $y \in Y$.

\textbf{If $\delta'$ is fixed}, then we just compute a lower bound on $\underline{\chi}_y$, which is $f_{\eta}(c(t) + d(\delta'))$, where $d(\delta')$ satisfying $\int^{d(\delta')}_{-d(\delta')} f_{\eta}(z)dz = \delta'$ can be computed as shown in Sec.~\ref{sec:errorbound}. After evaluating $\underline{\chi}_y$, the value of $\epsilon'$ can be found in a straightforward way as
\[\epsilon' = \frac{1}{1 + \frac{1}{e^{\epsilon} + \frac{\delta}{f_{\eta}(c(t) + d(\delta))}}\cdot\frac{q}{p}} - p\enspace,\]
\noindent where $\epsilon' < 0$ means that we could not satisfy $(\epsilon',\delta')$-GA for the given $\delta'$.

\textbf{If $\epsilon'$ is fixed}, then we need
\[\underline{\chi}_y \geq \frac{\delta}{\frac{p+\epsilon'}{1 - (p + \epsilon')} \cdot \frac{q}{p} - e^{\epsilon} }\enspace,\]
\noindent so we find $d$ such that $f_{\eta}(c(t) + d) = \underline{\chi}_y$ and compute $\delta' = \int^d_{-d} f_{\eta}(z)dz$. If $f_{\eta}(c(t)) \leq \frac{\delta}{\frac{p+\epsilon'}{1 - (p + \epsilon')} \cdot \frac{q}{p} - e^{\epsilon} }$, then such $d$ does not exist, meaning that we could not satisfy $(\epsilon',\delta')$-GA for the given $\epsilon'$.

\subsection{$(\epsilon,\delta)$-GA to $(\epsilon,\delta)$-DP}\label{sec:gatodp}

For our privacy analyses, we need to compute the converse: given an $(\epsilon',\delta')$ requirement on GA, find an appropriate $(\epsilon,\delta)$ for DP and generate the noise accordingly. The input can be given in terms of average guessing advantage $\Delta = 1 + (\epsilon' - 1) \cdot \delta'$, which gives freedom in choosing $\epsilon'$ and $\delta'$. We show that, as far as we can choose $\epsilon$ and $\delta$ ourselves, for some noise distributions (like Laplace) we can get $\delta' = 0$, which simplifies the solution and allows us to continue using the old definition of GA.

While intuitively $(\epsilon,\delta)$-DP means that pure $\epsilon$-DP is satisfied except some bad $y \in Y$ comes into play, this ``bad'' event probability is not necessarily $\delta$. There are different ways to define $(\epsilon,\delta)$-DP, discussed e.g. in~\cite{DBLP:journals/popets/SommerMM19}. We give transformations between DP and GA for the two main definitions.

\begin{definition}[Approximate DP~\cite{DBLP:journals/popets/SommerMM19}]\label{def:adp}
Let $\noised{q} : X \to Y$ be a probabilistic mechanism. We say $\noised{q}$ is $(\epsilon,\delta)$ differentially private if for all sets $S \subseteq Y$ and $x_0, x_1 \in X$ we have
\[ \pr{\noised{q}(x_0) \in S} \leq e^{\epsilon\cdot d(x_0,x_1)}\cdot\pr{\noised{q}(x_1) \in S} + \delta\enspace.\]
\end{definition}

\begin{definition}[Probabilistic DP~\cite{DBLP:journals/popets/SommerMM19}]\label{def:pdp}
Let $\noised{q} : X \to Y$ be a probabilistic mechanism. We say $\noised{q}$ is $(\epsilon,\delta)$-\emph{probabilistically} differentially private if for all $x_0, x_1 \in X$ there are sets $S^{\delta}_0, S^{\delta}_1 \subseteq Y$ with $\pr{\noised{q}(x_0) \in S^{\delta}_0} \leq \delta$ and $\pr{\noised{q}(x_1) \in S^{\delta}_1} \leq \delta$ s.t. for all sets $S \subseteq Y$ the following holds:
\begin{itemize}
\item $\pr{\noised{q}(x_0) \in S \setminus S^{\delta}_0} \leq e^{\epsilon\cdot d(x_0,x_1)}\cdot\pr{\noised{q}(x_1) \in S \setminus S^{\delta}_0}\enspace;$
\item $\pr{\noised{q}(x_1) \in S \setminus S^{\delta}_1} \leq e^{\epsilon\cdot d(x_0,x_1)}\cdot\pr{\noised{q}(x_0) \in S \setminus S^{\delta}_1}\enspace.$
\end{itemize}
\end{definition}

\paragraph{Achieving $(\epsilon,\delta)$-GA for Probabilistic DP.} It is easier to achieve $(\epsilon,\delta)$-GA for Definition~\ref{def:pdp}. Let $y \in Y$ be such that $\pr{y|x} \leq \pr{y|x'}e^{\epsilon}$ for all $x,x'$. By Definition~\ref{def:pdp}, this happens for at least $(1-\delta)$ of values $y \in Y$, so we now use the same $\delta$ in GA and DP definitions. For those $y$, we compute $\epsilon$ in the same way as we did it for $\epsilon$-DP. The desired guessing advantage is achieved with probability $1 - \delta$.

\paragraph{Achieving $(\epsilon,\delta)$-GA for Approximate DP.}  For Definition~\ref{def:adp}, we do not have a direct conversion to probability of failure, and we need to compute it ourselves, by estimating the noise distribution. Let $y \in Y$ be arbitrary. Applying definition of $(\epsilon,\delta)$-DP, we get $\underline{\chi}_y = \inf_{x \in \hat{X'}}f_Y(y|x) \geq (\sup_{x \in \hat{X'}}f_Y(y|x) -\delta)\cdot e^{-\epsilon}$, where $\sup_{x \in \hat{X'}}f_Y(y|x) =: \overline{\chi}_y$ is a strictly positive quantity.

\begin{eqnarray*}
\prpost{X'} & \leq & \frac{1}{1 + \frac{1}{e^{\epsilon} + \frac{\delta}{\inf_{x \in \hat{X'}}f_Y(y|x)}}\cdot\frac{q}{p}}\\
& \leq & \frac{1}{1 + \frac{1}{e^{\epsilon} + \frac{\delta}{e^{-\epsilon}(\sup_{x \in \hat{X'}}f_Y(y|x) - \delta)}}\cdot\frac{q}{p}}\\
& = & \frac{1}{1 + \frac{1}{e^{\epsilon} + \frac{\delta}{e^{-\epsilon}(\overline{\chi}_y - \delta)}}\cdot\frac{q}{p}} = \frac{1}{1 + e^{-\epsilon}\cdot\frac{1}{1 + \frac{\delta}{\overline{\chi}_y - \delta}}\cdot\frac{q}{p}}\\
& = & \frac{1}{1 + e^{-\epsilon}(1 - \frac{\delta}{\overline{\chi}_y})\cdot\frac{q}{p}}\enspace.
\end{eqnarray*}

If $\delta \geq \overline{\chi}_y$, then we cannot come up with $\epsilon > 0$, so we need $\delta$ to be small enough. The problem here is that $\epsilon$ and $\delta$ are correlated in a certain way. If we fixed $\delta = 0$, then we could find $\epsilon$ in the same way as we did it in D1.7, but it may happen that $(\epsilon,\delta)$-DP can be achieved only if infinite noise is used. The problem is that $\epsilon$ in turn depends on $\delta$, and to optimize them simultaneously, we need to know more about their correlation. Also, the quantity $\overline{\chi}_y$ depends on both $\epsilon$ and $\delta$, and we need to know the noise distribution to find a non-trivial lower bound on $\overline{\chi}_y$.

%If the noise is Laplace, using the fact that $\delta \leq 1$ gives us
%\[\frac{1}{1 + e^{-\epsilon}(1 - \frac{2}{\epsilon})\cdot\frac{q}{p}}\enspace,\]
%which does not allow to pick $\epsilon \leq 2$ and may fail to achieve desired guessing advantage.

Using results of~\cite{dersens}, for Laplace noise based on $\beta$-smooth derivative sensitivity, we need to take $\delta = 2e^{\epsilon-1-\frac{\epsilon-b}{\beta}}$, where $b + \beta \leq \epsilon$. For a fixed $b$, we get $\overline{\chi}_y = \frac{b}{2\cdot c_{\beta}(t)}$, where $c_{\beta}(t)$ is $\beta$-smooth derivative sensitivity of the query. We get
\[\frac{1}{1 + e^{-\epsilon}(1 - \frac{4e^{\epsilon-1 - \frac{\epsilon - b}{\beta}}\cdot c_{\beta}(t)}{b})\cdot\frac{q}{p}}\enspace.\]
To get a reasonable $\epsilon$, we need $(1 - \frac{4e^{\epsilon-1 - \frac{\epsilon - b}{\beta}}\cdot c_{\beta}(t)}{b}) > 0$. If we take $\beta \leq \frac{\epsilon - b}{\epsilon - 1 - \ln(\frac{b}{4})}$ (note that $\beta > 0$ as $b < \epsilon$), then we get

\[\frac{4e^{\epsilon-1 - \frac{\epsilon - b}{\beta}}\cdot c_{\beta}(t)}{b} \leq \frac{4e^{\ln(\frac{b}{4})}\cdot c_{\beta}(t)}{b} = c_{\beta}(t)\enspace.\]

If $c_{\beta}(t) < 1$, then we are done. In general, we can proceed and decrease $\beta$ as much as we like. If $c_{\beta}(t) \leq C$ for some constant $C \leq \infty$. we need to take $\beta \leq \frac{\epsilon - b}{\epsilon - 1 - \ln(\frac{b}{4C})}$. If there is no global upper bound on derivative sensitivity, we can see it as $C = \infty$, which gives $\beta = 0$ and results in infinite noise. In the latter case, there is still possibility that we can get finite noise, and a possible approach is to try out different values of $C'$ using window binary search and checking whether obtained $\beta$ satisfies $c_{\beta}(t) \leq C'$. This can be useful also if $C < \infty$, just to check if a better solution exists.

Since there are several parameters to optimize, we now describe a particular procedure of finding appropriate parameters.
\begin{itemize}
\item Let $\epsilon$ be the parameter we would get from $\epsilon$-DP (i.e. for $\delta = 0$).
%\item Let $\mu$ be a parameter that relates similarity of $\epsilon$ that we would get with $\delta = 0$ to the $\epsilon$ that we will be able to get for $\delta > 0$. E.g. $\mu = 0.99$. Assume that $e^{-\epsilon}$ in GA formula is multiplied by some $\alpha < 1$. Let $\alpha = e^{-(1 - \mu)\epsilon}$. This gives us the result $-\epsilon + (1 - \mu)\epsilon = -\mu\epsilon$. I.e. if $\epsilon$ would be the parameter we got for $\epsilon$-GA, we would need $\mu\epsilon$ for $(\epsilon,\delta)$-GA.

\item Let $\beta$ and $b$ be the optimal values for \emph{Cauchy noise}, found using search over $\beta$ as described in~\cite{dersens}. Note that these will not necessary be optimal for \emph{Laplace noise}, but they give a fair starting point.

\item Assume that $e^{-\epsilon}$ in GA formula is multiplied by some $e^{-\alpha} \leq 1$. This means that, if $\epsilon$ would be the parameter we got for $\epsilon$-DP, we would need $\epsilon' = \epsilon - \alpha$ for $(\epsilon',\delta)$-DP. This sets a constraint $\alpha \leq \epsilon$. Now $\alpha \in [0,\epsilon]$ is a parameter that can be optimized.

\item We need to find $b'$ such that $1 - \frac{4e^{\epsilon-\alpha-1 - \frac{\epsilon-\alpha-b'}{\beta}}\cdot c_{\beta}(t)}{b'} \geq e^{-\alpha\epsilon}$. For this, we need to satisfy the equation $\frac{e^{b'/\beta}}{b'/\beta} \leq \frac{(1-e^{-\alpha})\beta}{4c_{\beta}(t)\cdot e^{\epsilon-\alpha-1-(\epsilon-\alpha)/\beta}}$. The solution that satisfies the \emph{equality} is $b' = - \beta\cdot W(-\frac{4c_{\beta}(t)\cdot e^{\epsilon-\alpha-1-(\epsilon-\alpha)/\beta}}{(1-e^{-\alpha})\beta})$, where $W$ is Lambert's $W$ function. By definition, $W(y) = x$ such that $x \cdot e^{x} = y$, and since $\frac{e^{x}}{x} = \frac{-1}{e^{-x}(-x)}$, the equation $e^{x} / x = y$ can be solved as $x = - W(-\frac{1}{y})$. If $\beta = 0$, we will have division by $0$, but in this case we can compute directly $\delta = 0$, which gives us $\alpha = 0$ and $b' = \epsilon - \beta$.

While using Lambert's W gives a correct solution, it does not necessarily give us the best noise. We need to satisfy the \emph{inequality} $1 - \frac{4e^{\epsilon-\alpha-1 - \frac{\epsilon-\alpha-b'}{\beta}}\cdot c_{\beta}(t)}{b'} \geq e^{-\alpha\epsilon}$, and we may get better results for values of $b$ that make this inequality non-strict. Hence, it makes sense to optimize the parameter $b'$ by searching through $b' \in [0,\ldots,\epsilon-\alpha-\beta]$.

\item So far, $\alpha$ is a free parameter that can be optimized. We do it by searching through $\alpha \in [0,\epsilon]$.

\item The noise level $\frac{c_{\beta}(t)}{b'}$ that we get in this way gives us the desired bound on guessing advantage.
\end{itemize}

\end{document}

%% file: corner.pdf_t
\begin{picture}(0,0)%
\includegraphics{corner.pdf}%
\end{picture}%
\setlength{\unitlength}{4144sp}%
\begingroup\makeatletter\ifx\SetFigFont\undefined%
\gdef\SetFigFont#1#2#3#4#5{%
  \reset@font\fontsize{#1}{#2pt}%
  \fontfamily{#3}\fontseries{#4}\fontshape{#5}%
  \selectfont}%
\fi\endgroup%
\begin{picture}(1374,1374)(-11,-523)
\put(766,434){\makebox(0,0)[lb]{\smash{{\SetFigFont{12}{14.4}{\rmdefault}{\mddefault}{\updefault}{\color[rgb]{0,0,0}$a$}%
}}}}
\put(1036,-376){\makebox(0,0)[lb]{\smash{{\SetFigFont{12}{14.4}{\rmdefault}{\mddefault}{\updefault}{\color[rgb]{0,0,0}$X'$}%
}}}}
\put(226,-286){\makebox(0,0)[lb]{\smash{{\SetFigFont{12}{14.4}{\rmdefault}{\mddefault}{\updefault}{\color[rgb]{0,0,0}$X$}%
}}}}
\end{picture}%

%% file: andor2D.pdf_t
\begin{picture}(0,0)%
\includegraphics{andor2D.pdf}%
\end{picture}%
\setlength{\unitlength}{4144sp}%
\begingroup\makeatletter\ifx\SetFigFont\undefined%
\gdef\SetFigFont#1#2#3#4#5{%
  \reset@font\fontsize{#1}{#2pt}%
  \fontfamily{#3}\fontseries{#4}\fontshape{#5}%
  \selectfont}%
\fi\endgroup%
\begin{picture}(2772,2454)(-11,-1648)
\put(1126,659){\makebox(0,0)[lb]{\smash{{\SetFigFont{12}{14.4}{\rmdefault}{\mddefault}{\updefault}{\color[rgb]{0,0,0}$x_i - t_i \leq r_i$}%
}}}}
\put(2746,-556){\makebox(0,0)[lb]{\smash{{\SetFigFont{12}{14.4}{\rmdefault}{\mddefault}{\updefault}{\color[rgb]{0,0,0}$x_j - t_j \leq r_j$}%
}}}}
\put(946,-691){\makebox(0,0)[lb]{\smash{{\SetFigFont{12}{14.4}{\rmdefault}{\mddefault}{\updefault}{\color[rgb]{0,0,0}$r_j$}%
}}}}
\put(1441,-376){\makebox(0,0)[lb]{\smash{{\SetFigFont{12}{14.4}{\rmdefault}{\mddefault}{\updefault}{\color[rgb]{0,0,0}$r_i$}%
}}}}
\put(1306,-781){\makebox(0,0)[lb]{\smash{{\SetFigFont{12}{14.4}{\rmdefault}{\mddefault}{\updefault}{\color[rgb]{0,0,0}$(t_i,t_j)$}%
}}}}
\end{picture}%

%% file: split.pdf_t
\begin{picture}(0,0)%
\includegraphics{split.pdf}%
\end{picture}%
\setlength{\unitlength}{4144sp}%
\begingroup\makeatletter\ifx\SetFigFont\undefined%
\gdef\SetFigFont#1#2#3#4#5{%
  \reset@font\fontsize{#1}{#2pt}%
  \fontfamily{#3}\fontseries{#4}\fontshape{#5}%
  \selectfont}%
\fi\endgroup%
\begin{picture}(2322,2454)(-11,-1648)
\put(2296,-556){\makebox(0,0)[lb]{\smash{{\SetFigFont{12}{14.4}{\rmdefault}{\mddefault}{\updefault}{\color[rgb]{0,0,0}$x_j - t_j \leq r_j$}%
}}}}
\put(856,659){\makebox(0,0)[lb]{\smash{{\SetFigFont{12}{14.4}{\rmdefault}{\mddefault}{\updefault}{\color[rgb]{0,0,0}$x_i - t_i \leq r_i$}%
}}}}
\put(991,-151){\makebox(0,0)[lb]{\smash{{\SetFigFont{12}{14.4}{\rmdefault}{\mddefault}{\updefault}{\color[rgb]{0,0,0}$a$}%
}}}}
\put(1441,-466){\makebox(0,0)[lb]{\smash{{\SetFigFont{12}{14.4}{\rmdefault}{\mddefault}{\updefault}{\color[rgb]{0,0,0}$a$}%
}}}}
\put(1036,-691){\makebox(0,0)[lb]{\smash{{\SetFigFont{12}{14.4}{\rmdefault}{\mddefault}{\updefault}{\color[rgb]{0,0,0}$(t_i,t_j)$}%
}}}}
\put(586,-556){\makebox(0,0)[lb]{\smash{{\SetFigFont{12}{14.4}{\rmdefault}{\mddefault}{\updefault}{\color[rgb]{0,0,0}$X'_0$}%
}}}}
\put(586,-151){\makebox(0,0)[lb]{\smash{{\SetFigFont{12}{14.4}{\rmdefault}{\mddefault}{\updefault}{\color[rgb]{0,0,0}$\hat{X'}_0$}%
}}}}
\put(1441,-151){\makebox(0,0)[lb]{\smash{{\SetFigFont{12}{14.4}{\rmdefault}{\mddefault}{\updefault}{\color[rgb]{0,0,0}$\hat{X'}_0$}%
}}}}
\put(586,-1096){\makebox(0,0)[lb]{\smash{{\SetFigFont{12}{14.4}{\rmdefault}{\mddefault}{\updefault}{\color[rgb]{0,0,0}$\hat{X'}_0$}%
}}}}
\put(1441,-1096){\makebox(0,0)[lb]{\smash{{\SetFigFont{12}{14.4}{\rmdefault}{\mddefault}{\updefault}{\color[rgb]{0,0,0}$\hat{X'}_0$}%
}}}}
\put(1936,-556){\makebox(0,0)[lb]{\smash{{\SetFigFont{12}{14.4}{\rmdefault}{\mddefault}{\updefault}{\color[rgb]{0,0,0}$X'_1$}%
}}}}
\put(1081,299){\makebox(0,0)[lb]{\smash{{\SetFigFont{12}{14.4}{\rmdefault}{\mddefault}{\updefault}{\color[rgb]{0,0,0}$X'_2$}%
}}}}
\put(991,-1501){\makebox(0,0)[lb]{\smash{{\SetFigFont{12}{14.4}{\rmdefault}{\mddefault}{\updefault}{\color[rgb]{0,0,0}$X'_4$}%
}}}}
\put( 91,-556){\makebox(0,0)[lb]{\smash{{\SetFigFont{12}{14.4}{\rmdefault}{\mddefault}{\updefault}{\color[rgb]{0,0,0}$X'_3$}%
}}}}
\put( 91,-1096){\makebox(0,0)[lb]{\smash{{\SetFigFont{12}{14.4}{\rmdefault}{\mddefault}{\updefault}{\color[rgb]{0,0,0}$\hat{X'}_3$}%
}}}}
\put( 91,-151){\makebox(0,0)[lb]{\smash{{\SetFigFont{12}{14.4}{\rmdefault}{\mddefault}{\updefault}{\color[rgb]{0,0,0}$\hat{X'}_3$}%
}}}}
\put(541,299){\makebox(0,0)[lb]{\smash{{\SetFigFont{12}{14.4}{\rmdefault}{\mddefault}{\updefault}{\color[rgb]{0,0,0}$\hat{X'}_2$}%
}}}}
\put(1891,-151){\makebox(0,0)[lb]{\smash{{\SetFigFont{12}{14.4}{\rmdefault}{\mddefault}{\updefault}{\color[rgb]{0,0,0}$\hat{X'}_1$}%
}}}}
\put(541,-1501){\makebox(0,0)[lb]{\smash{{\SetFigFont{12}{14.4}{\rmdefault}{\mddefault}{\updefault}{\color[rgb]{0,0,0}$\hat{X'}_4$}%
}}}}
\put(1441,299){\makebox(0,0)[lb]{\smash{{\SetFigFont{12}{14.4}{\rmdefault}{\mddefault}{\updefault}{\color[rgb]{0,0,0}$\hat{X'}_2$}%
}}}}
\put(1891,-1096){\makebox(0,0)[lb]{\smash{{\SetFigFont{12}{14.4}{\rmdefault}{\mddefault}{\updefault}{\color[rgb]{0,0,0}$\hat{X'}_1$}%
}}}}
\put(1441,-1501){\makebox(0,0)[lb]{\smash{{\SetFigFont{12}{14.4}{\rmdefault}{\mddefault}{\updefault}{\color[rgb]{0,0,0}$\hat{X'}_4$}%
}}}}
\end{picture}%